\documentclass[11pt,a4paper]{article}
\usepackage[a4paper, total={6in, 8in}]{geometry}
\usepackage{amssymb}
\usepackage{makeidx}
\usepackage{amsfonts}
\usepackage{amsmath}
\usepackage{amsthm}
\usepackage{graphicx}
\usepackage[dvips]{epsfig}
\usepackage{longtable}
\usepackage{graphics}
\usepackage{multirow}
\usepackage[section]{placeins}

\usepackage{graphicx}				
\usepackage{amssymb}
\usepackage{mathtools}
\usepackage{tensor}

\usepackage{todonotes}

\newtheorem{theorem}{Theorem}
\newtheorem{lemma}{Lemma}
\newtheorem{corollary}{Corollary}

\newtheorem{remark}{Remark}

\newcommand{\trans}{{\mathrm{T}}}

\newcommand{\E}{\operatorname{E}}
\renewcommand{\P}{\operatorname{P}}

\newcommand{\var}{\operatorname{var}}
\newcommand{\cov}{\operatorname{cov}}

\renewcommand{\phi}{\varphi}

\newcommand{\eqdis}{\stackrel{d}{=}}

\newcommand{\distr}{\stackrel{d}{\to}}
\newcommand{\as}{\stackrel{\textrm{a.s.}}{\to}}

\newcommand{\diag}{\operatorname{diag}}

\newcommand{\Pois}{\operatorname{Pois}}

\newcommand{\Normal}{\operatorname{N}}

\title{Mack's estimator motivated by large exposure asymptotics in a compound Poisson setting}
\author{Nils Engler\footnote{nils.engler@math.su.se, Department of Mathematics, Stockholm University, Sweden} \, and Filip Lindskog\footnote{Corresponding author, lindskog@math.su.se, Department of Mathematics, Stockholm University, Sweden}}

\begin{document}
\maketitle

\begin{abstract}
The distribution-free chain ladder of Mack justified the use of the chain ladder predictor and enabled Mack to derive an estimator of  conditional mean squared error of prediction for the chain ladder predictor. Classical insurance loss models, i.e.~of compound Poisson type, are not consistent with Mack's distribution-free chain ladder. 
However, for a sequence of compound Poisson loss models indexed by exposure (e.g.~number of contracts), we show that the chain ladder predictor and Mack's estimator of conditional mean squared error of prediction can be derived by considering large exposure asymptotics. Hence, quantifying chain ladder prediction uncertainty can be done with Mack's estimator without relying on the validity of the model assumptions of the distribution-free chain ladder.   
\end{abstract}

\noindent {\bf Keywords}: claims reserving, chain ladder, large exposure asymptotics

\section{Introduction}

We consider the problem of predicting outstanding claims costs from insurance contracts whose coverage periods have expired but for which not all claims are known to the insurer. Such prediction tasks are referred to as claims reserving.  
The chain ladder method is arguably the most widespread and well known technique for claims reserving based on claims data organized in run-off triangles, with cells indexed by accident year and deve\-lopment year. The chain ladder method is a deterministic prediction method for predicting the not yet known south-east corner (target triangle) based on the observed north-west corner (historical triangle) of a square with cell values representing accumulated total claims amounts. 
The square and historical triangle can easily be generalized to rectangle and trapezoid, reflecting claims data for more historical accident years. However, we will here consider the traditional setup in order to simplify comparison with influential papers.  
We refer to the text book \cite{Wuthrich-Merz-08} by W\"uthrich and Merz for an overview of methods for claims reserving.  

Important contributions appeared in the 1990s presenting stochastic models and properties of parametric stochastic models that give rise to the chain ladder predictor. Mack \cite{Mack-93} presented three model properties, known as the distribution-free chain ladder model, that together with weighted least squares estimation give rise to the chain ladder predictor. Renshaw and Verrall \cite{Renshaw-Verrall-98} showed that independent Poisson distributed cell values for incremental total claims amounts, together with Maximum Likelihood estimation of parameters for row and column effects, give rise to the chain ladder predictor. The Poisson model is inconsistent with the distribution-free chain ladder. 

The most impressive contribution of Mack in \cite{Mack-93} is the estimator of conditional mean squared error of prediction. The key contribution is the estimator of the contribution of parameter estimation error to conditional mean squared error of prediction. A number of papers have derived the same estimator based on different approaches to statistical estimation in settings consistent with the distribution-free chain ladder, see e.g.~Merz and W\"uthrich \cite{Merz-Wuthrich-08}, R\"ohr \cite{Rohr-16}, Diers et al.~\cite{Diers-et-al-16}, Gisler \cite{Gisler-19}, Lindholm et al.~\cite{Lindholm-et-al-19}. 

Different approaches to the estimation of, and estimators of, prediction error for the chain ladder method sparked some scientific debate, both regarding which stochastic model underlies the chain ladder methods, see e.g.~the papers by Mack and Venter \cite{Mack-Venter-00} and Verrall and England \cite{Verrall-England-00}, and regarding prediction error estimation for the chain ladder method, see 
Buchwalder et al.~\cite{Buchwalder-et-al-06}, Gisler \cite{Gisler-06}, Mack et al.~\cite{Mack-Quarg-Braun-06} and Venter \cite{Venter-06}. 
Gisler revisited, in \cite{Gisler-21}, different estimators for conditional mean squared error in the setting of the distribution-free chain ladder. 
Ultimately, Mack's estimator of conditional mean squared error of prediction has stood the test of time. 

The main contribution of the present paper is that we show that a simple but natural compound Poisson model is fully compatible with both the chain ladder predictor and Mack's estimator of conditional mean squared error of prediction, although the model is incompatible with Mack's distribution-free chain ladder, as long as we consider an insurance portfolio with sufficiently large exposure (e.g.~accumulated total claims amounts based on sufficiently  many contracts). The Poisson model considered by Renshaw and Verrall in \cite{Renshaw-Verrall-98} is a special case of the compound Poisson model we consider, and consequently also their Poisson model gives rise to Mack's estimator of conditional mean squared error of prediction. 
     
The rest of the paper is organized as follows. Section \ref{sec:themodel} presents the stochastic model we consider, both a simple model called the special model and a more general model. The special model is a classical insurance loss model (independent compound Poisson processes in each cell of the run-off triangle of incremental total claim amounts).      
Section \ref{sec:dfcl} recalls Mack's distribution-free chain ladder. 
Section \ref{sec:lea} presents asymptotic results that demonstrates that we can retrieve Mack's classical estimators in model setting that are incompatible with the distribution-free chain ladder. 
Section \ref{sec:numeric_example} presents a numerical example that illustrates the theoretical results in Section \ref{sec:lea}.  
The proofs are found in Section \ref{sec:proofs}. 

\section{The model}\label{sec:themodel}

We will focus on a simple yet general class of models for the number of reported claims and the cost of these claims. 
In line with classical reserving methods based on claims data organized in run-off triangles, we consider $T$ accident years and $T$ development years. For $i,t\in \mathcal{T}=\{1,\dots,T\}$, let $C^{\alpha}_{i,t}$ denote the accumulated total claims amount due to accident events in accident year $i$ that are paid up to and including development year $t$. The parameter $\alpha$ is a measure of exposure, such as the number of contracts of not yet fully developed accident years. We will analyze asymptotics as $\alpha\to\infty$ and use the findings to motivate the use of well established predictors and estimators in settings that are not consistent with model assumptions used to derive the classical results for the chain ladder method. A given claims reserving situation of course corresponds to a single, typically large, number $\alpha$. As in any other situation where asymptotic arguments are the basis for approximation, we embed the prediction problem in a sequence of prediction problems, indexed by $\alpha$.    

The {\bf special model} is simply a set of independent Cram\'er-Lundberg (compound Poisson) models, indexed by accident year and development year, with a common claim size distribution with finite variance and positive mean, where exposure parameter $\alpha$ plays the role of time in the Cram\'er-Lundberg models.    
Consider incremental accumulated total claim amounts $X^{\alpha}_{i,t}$ due to accident events in accident year $i$ that are paid during development year $t$: $X^{\alpha}_{i,1}=C^{\alpha}_{i,1}$ and $X^{\alpha}_{i,t}=C^{\alpha}_{i,t}-C^{\alpha}_{i,t-1}$ for $t\geq 2$. 
Consider constants $\lambda_1,\dots,\lambda_T\in (0,\infty)$ and $q_1,\dots,q_T\in (0,1)$ with $\sum_{t=1}^{T}q_t=1$. 
For each $i,t\in \mathcal{T}$, $(X^{\alpha}_{i,t})_{\alpha\geq 0}$ is a Cram\'er-Lundberg model with representation 
$$
X^{\alpha}_{i,t}=\sum_{k=1}^{N^{\alpha}_{i,t}}Z_{i,t,k}, \quad \alpha\geq 0,
$$ 
where $(N^{\alpha}_{i,t})_{\alpha\geq 0}$ is a homogeneous Poisson process with intensity $\lambda_iq_t\in (0,\infty)$, independent of the i.i.d.~sequence $(Z_{i,t,k})_{k=1}^{\infty}$. The claim size variables satisfy $Z_{i,t,k}\eqdis Z$ for all $i,t,k$ for some $Z$ with finite variance and positive mean. 
Moreover, the compound Poisson processes $(X^{\alpha}_{i,t})_{\alpha\geq 0}$, $(i,t)\in \mathcal{T}\times \mathcal{T}$, are independent.   

We want to highlight the special case of the special model obtained by letting $Z\equiv 1$. In this case the special model is simply a set of independent homogeneous Poisson processes, indexed by accident year and development year, where exposure parameter $\alpha$ plays the role of time. In particular, for a fixed $\alpha$, we obtain the model considered by Renshaw and Verrall in \cite{Renshaw-Verrall-98} as a model underlying the chain ladder method since it gives rise to the chain ladder predictor (see Section \ref{sec:dfcl}) upon replacing unknown parameters by their Maximum Likelihood estimates. 

\subsection{The general model}

Several of the statements in Section \ref{sec:lea} below hold for a wider class of models than the special model. 
The {\bf general model}, (GM1)-(GM4) below, allows us to write 
\begin{align*}
C^{\alpha}_{i,t}=\sum_{k=1}^{M^{\alpha}_i}Z_{i,k} I\{D_{i,k}\leq t\},
\end{align*}  
where $M^{\alpha}_i$ denotes the number of accident events in accident year $i$, $Z_{i,k}$ denotes the size of the $k$th such claim and $D_{i,k}$ denotes the corresponding development year, the indicator $I\{D_{i,k}\leq t\}$ equals $1$ if $D_{i,k}\leq t$. 
The properties GM1-GM4 together constitute the {\bf general model}: 
\begin{itemize}
\item[(GM1)] $(D_{1,k},Z_{1,k})_{k=1}^{\infty}, \dots, (D_{T,k},Z_{T,k})_{k=1}^{\infty}$ are i.i.d.~sequences. The common distribution of the 
terms $(D_{i,k},Z_{i,k})$ does not depend on the accident year $i$. With $(D,Z)$ denoting a generic such pair, 
\begin{align*}
\E[Z^2]<\infty \quad\text{and}\quad \E[ZI\{D=t\}]>0 \quad\text{for each } t\in\mathcal{T}. 
\end{align*} 
\item[(GM2)] For each $i$, $(D_{i,k},Z_{i,k})_{k=1}^{\infty}$ and $M^{\alpha}_i$ are independent.
\item[(GM3)] $\{M^{\alpha}_1,(D_{1,k},Z_{1,k})_{k=1}^{\infty}\}, \dots, \{M^{\alpha}_T,(D_{T,k},Z_{T,k})_{k=1}^{\infty}\}$ are independent. 
\item[(GM4)] For each $i$ there exists $\lambda_i\in (0,\infty)$ such that $M^{\alpha}_i/\alpha \stackrel{\textrm{a.s.}}{\to} \lambda_i$ as $\alpha\to\infty$. 
\end{itemize}
By (GM3), claims data variables are independent if they correspond to different accident years. However, the components of $(D,Z)$ are possibly dependent, allowing for the distribution of claim size to depend on development year. Note that we allow for exposures to vary between accident years, reflected in possibly different parameters $\lambda_1,\dots,\lambda_T$. Note also that the incremental accumulated claims amounts $X^{\alpha}_{i,s}$ and $X^{\alpha}_{i,t}$, $s\neq t$, are in general not independent (unless $M^{\alpha}_i$ is Poisson distributed).   

In order to derive Mack's estimator in \cite{Mack-93} of conditional mean squared error of prediction for the chain ladder predictor we must consider a special case of the general model:   
\begin{itemize} 
\item[(SM1)] (GM1)-(GM3) hold.
\item[(SM2)] $D$ and $Z$ are independent. 
\item[(SM3)] For each $i$, $(M^{\alpha}_i)_{\alpha\geq 0}$ is a homogeneous Poisson process with intensity $\lambda_i\in (0,\infty)$. 
\end{itemize}
The properties (SM1)-(SM3) together form an alternative way of specifying the {\bf special model}. Since (SM3) implies (GM4), the special model is a special case of the general model. 

\section{Mack's distribution-free chain ladder}\label{sec:dfcl}

The arguably most well-known method for claims reserving is the chain ladder method. In the seminal paper \cite{Mack-93}, Thomas Mack presented properties, see \eqref{eq:MackCL_cond1} and \eqref{eq:MackCL_cond2} below, for conditional distributions of accumulated total claims amounts that, together with \eqref{eq:MackCL_cond3} below, make the chain ladder prediction method the optimal prediction method for predicting outstanding claims amounts. 
Moreover, and this is the main contribution of \cite{Mack-93}, he showed that these properties lead to an estimator of the conditional mean squared error of the chain ladder predictor.  

With $C_{i,t}$ denoting the accumulated total claims amount up to and including development year $t$ for accidents during accident year $i$, Mack considered the following assumptions for the data generating process: for $t=1,\dots,T-1$ there exist constants $f_{\textrm{MCL}t}>0$ and $\sigma_{\textrm{MCL}t}^2\geq 0$ such that 
\begin{align}
&\E[C_{i,t+1}\mid C_{i,1},\dots,C_{i,t}]=f_{\textrm{MCL}t}C_{i,t}, \quad t=1,\dots,T-1, \label{eq:MackCL_cond1} \\ 
&\var(C_{i,t+1}\mid C_{i,1},\dots,C_{i,t})=\sigma_{\textrm{MCL}t}^2C_{i,t}, \quad t=1,\dots,T-1, \label{eq:MackCL_cond2}
\end{align}
and 
\begin{align}\label{eq:MackCL_cond3}
(C_{1,1},\dots,C_{1,T}), \dots, (C_{T,1},\dots,C_{T,T}) \quad\text{are independent}. 
\end{align}
The conditions \eqref{eq:MackCL_cond1}, \eqref{eq:MackCL_cond2} and \eqref{eq:MackCL_cond3} together are referred to as Mack's distribution-free chain ladder model. 
The parameters $f_{\textrm{MCL}t}$ and $\sigma^2_{\textrm{MCL}t}$ are estimated by 
\begin{align*}
\widehat{f}_{t}=\frac{\sum_{i=1}^{T-t}C_{i,t+1}}{\sum_{i=1}^{T-t}C_{i,t}}
\quad \text{and} \quad 
\widehat{\sigma}^2_{t}=\frac{1}{T-t-1}\sum_{i=1}^{T-t}C_{i,t}\bigg(\frac{C_{i,t+1}}{C_{i,t}}-\widehat{f}_{t}\bigg)^2,
\end{align*}
respectively. We refer to \cite{Mack-93} for properties of these parameter estimators. 

The property \eqref{eq:MackCL_cond2} for the conditional variance is very difficult to assess from data in the form of run-off triangles on which the chain ladder method is applied. We refer to \cite{Mack-94-CAS} for tests assessing the assumptions of Mack's distribution-free chain ladder. Moreover, it is notoriously difficult to find stochastic models that satisfy this property. Note that the special model, see Section \ref{sec:themodel}, does not satisfy Mack's conditions: neither \eqref{eq:MackCL_cond1} nor \eqref{eq:MackCL_cond2} hold. By Theorem 3.3.6. in \cite{Mikosch-09}, for the special model, 
\begin{align*}
\sum_{k=1}^{M^{\alpha}_i}Z_{i,k} I\{D_{i,k}\leq t\} \quad \text{and} \quad \sum_{k=1}^{M^{\alpha}_i}Z_{i,k} I\{D_{i,k}=t+1\}
\end{align*}
are independent. Consequently, for the special model, 
\begin{align*}
\E[C^{\alpha}_{i,t+1}\mid C^{\alpha}_{i,1},\dots,C^{\alpha}_{i,t}]=C^{\alpha}_{i,t}+\E[M^{\alpha}_i]\P(D=t+1)\E[Z]
\end{align*}
and 
\begin{align*}
\var(C^{\alpha}_{i,t+1}\mid C^{\alpha}_{i,1},\dots,C^{\alpha}_{i,t})=\E[M^{\alpha}_i]\P(D=t+1)\E[Z^2].
\end{align*}

It is shown in Theorem \ref{thm:estimation} below that large exposure limits, as $\alpha\to\infty$, do exist for estimators $\widehat{f}_t$ and $\widehat{\sigma}_t^2$. The constant (a.s.~convergence) limit for the parameter estimator $\widehat{f}_t$ has a meaningful interpretation in terms of the general model we consider, and the parameter estimators $\widehat{f}_t$ can be transformed into estimators of parameters of our model, see Remark \ref{rem:f_to_q}. However, Mack's parameter estimator $\widehat{\sigma}_t^2$ converges in distribution to a nondegenerate random variable. Hence, although $\widehat{\sigma}_t^2$ will generate numerical values that may seem reasonable, such values do not correspond to outcomes of random variables converging to a parameter. 

The main contribution of Mack's paper \cite{Mack-93} is the derivation of an estimator of the conditional mean squared error of prediction 
\begin{align*}
\E\big[(C_{i,T}-\widehat{C}_{i,T})^2 \mid \mathcal{D}\big], 
\end{align*}
where $\mathcal{D}$ is the $\sigma$-algebra generated by the data observed at the time of prediction: $\{C_{j,t}:j,t\in\mathcal{T},j+t\leq T+1\}$. The $\mathcal{D}$-measurable estimator derived by Mack of $\E[(C_{i,T}-\widehat{C}_{i,T})^2 \mid \mathcal{D}]$ is  
\begin{align}\label{eq:mack_msep_estimator} 
(\widehat{C}_{i,T})^2
\sum_{t=T-i+1}^{T-1}\frac{\widehat{\sigma}^2_t}{\widehat{f}^2_t}
\bigg(\frac{1}{\widehat{C}_{i,t}}+\frac{1}{\sum_{j=1}^{T-t}C_{j,t}}\bigg),
\end{align}
where $\widehat{C}_{i,T-i+1}=C_{i,T-i+1}$ and $\widehat{C}_{i,t}=C_{i,T-i+1}\prod_{s=T-i+1}^{t-1}\widehat{f}_s$ for $t>T-i+1$.  
We will show that when considering the special model (SM1)-(SM3), large exposure asymptotics naturally lead to Mack's estimator of conditional mean squared error of prediction despite the fact that the special model is inconsistent with Mack's distribution-free chain ladder.  
Hence, the chain ladder predictor $\widehat{C}_{i,T}=C_{i,T-i+1}\prod_{s=T-i+1}^{T-1}\widehat{f}_s$ may be used together with an assessment of its accuracy by \eqref{eq:mack_msep_estimator} without having to rely on the validity of \eqref{eq:MackCL_cond1} and \eqref{eq:MackCL_cond2} of Mack's distribution-free chain ladder. 

\section{Large exposure asymptotics}\label{sec:lea}

We will next present the main results, motivating the use of the chain ladder method and Mack's estimator of conditional mean squared error of prediction, in the setting of the general or special model. 
Recall that, for $i,t\in \mathcal{T}$, $C^{\alpha}_{i,t}=\sum_{k=1}^{M^{\alpha}_i}Z_{i,k}I\{D_{i,k}\leq t\}$. 
Let $\chi^2_{\nu}$ denote a random variable with a chi squared distribution with $\nu$ degrees of freedom. Let $\Normal_T(\mu,\Sigma)$ denote the $T$-dimensional normal distribution with mean $\mu$ and covariance matrix $\Sigma$. In what follows, convergence of random variables should be understood as convergence as $\alpha\to\infty$. 

\begin{theorem}\label{thm:estimation}
Consider the general model (GM1)-(GM4). 
For each $t\in\mathcal{T}$ with $t\leq T-1$,  
\begin{align}\label{eq:f_hat}
\widehat{f}_{t}=&\frac{\sum_{i=1}^{T-t}C^{\alpha}_{i,t+1}}{\sum_{i=1}^{T-t}C^{\alpha}_{i,t}}
\as \frac{\E[ZI\{D\leq t+1\}]}{\E[ZI\{D\leq t\}]}=f_{t}.
\end{align}
For each $i\in\mathcal{T}$ with $i\geq 2$, 
\begin{align}\label{eq:clpredictor}
\frac{C^{\alpha}_{i,T-i+1}\prod_{t=T-i+1}^{T-1}\widehat{f}_t}{C^{\alpha}_{i,T}} \as 1.
\end{align}
For each $t\in\mathcal{T}$ with $t\leq T-2$,  
\begin{align}\label{eq:sigma_hat}
\widehat{\sigma}^2_{t}=\frac{1}{T-t-1}\sum_{i=1}^{T-t}C^{\alpha}_{i,t}\bigg(\frac{C^{\alpha}_{i,t+1}}{C^{\alpha}_{i,t}}-\widehat{f}_{t}\bigg)^2 \distr 
\sigma^2_t \frac{\chi^2_{T-t-1}}{T-t-1},  
\end{align}
where
\begin{align*}
\sigma^2_t = (f_t-1)\bigg(\frac{\E[Z^2 I\{D=t+1\}]}{\E[Z I\{D=t+1\}]}+(f_t-1)\frac{\E[Z^2 I\{D\leq t\}]}{\E[Z I\{D\leq t\}]}\bigg).
\end{align*}
\end{theorem}

\begin{remark}
We do not index $\widehat{f}_{t}$ and $\widehat{\sigma}^2_{t}$ by the exposure parameter $\alpha$. It should be clear from the context whether $\widehat{f}_{t}$ should be seen as an element in a convergent sequence or simply as a function of the given data. Similarly for $\widehat{\sigma}^2_{t}$.  
\end{remark}

\begin{remark}
For the convergence in \eqref{eq:f_hat} and \eqref{eq:clpredictor} it is not necessary to assume that $M^{\alpha}_{1},\dots,M^{\alpha}_{T}$ are independent. 
If $Z$ and $D$ are independent, then the limit expressions in \eqref{eq:f_hat} and \eqref{eq:sigma_hat} simplify:   
\begin{align*}
f_t=\frac{\sum_{s=1}^{t+1}q_s}{\sum_{s=1}^{t}q_s}, \quad 
\sigma^2_t=(f_t-1)f_t\frac{\E[Z^2]}{\E[Z]}=\frac{q_{t+1}\sum_{s=1}^{t+1}q_s}{(\sum_{s=1}^{t}q_s)^2} \frac{\E[Z^2]}{\E[Z]},
\end{align*}
where $q_t=\P(D=t)$. 
\end{remark}

\begin{remark}
The convergence \eqref{eq:clpredictor} supports the use of the chain ladder predictor 
$$
\widehat{C}_{i,T}=C_{i,T-i+1}\widehat{f}_{T-i+1} \dots  \widehat{f}_{T-1}
$$ 
whose prediction error is studied in \cite{Mack-93} and \cite{Mack-94-CAS}.  
However, \eqref{eq:sigma_hat} says that from numerical estimates $\widehat{\sigma}^2_{t}$ we may not conclude that there is empirical evidence in support of the assumption \eqref{eq:MackCL_cond2} of Mack's distribution-free chain ladder.  
\end{remark}

\begin{remark}\label{rem:f_to_q}
It follows from \eqref{eq:f_hat} that if we either replace the claims amounts by the number of claims (corresponding to $Z\equiv 1$) in the estimator $\widehat{f}_t$, or assume that the variables $D$ and $Z$ are independent, then the estimators $\widehat{f}_1,\dots,\widehat{f}_{T-1}$ can be transformed into consistent estimators $\widehat{q}_1,\dots,\widehat{q}_{T}$, where $q_t=\P(D=t)$. 
More generally, 
\begin{align*}
\bigg(\frac{1}{\prod_{s=1}^{T-1}\widehat{f}_s},\frac{\widehat{f}_{1}-1}{\prod_{s=1}^{T-1}\widehat{f}_s},\frac{(\widehat{f}_{2}-1)\widehat{f}_1}{\prod_{s=1}^{T-1}\widehat{f}_s},\dots,\frac{(\widehat{f}_{T-1}-1)\prod_{s=1}^{T-2}\widehat{f}_s}{\prod_{s=1}^{T-1}\widehat{f}_s}\bigg)
\end{align*}
converges a.s.~to $(\widetilde{q}_1,\dots,\widetilde{q}_T)$, where $\widetilde{q}_t=\widetilde{\P}(D=t)=\E[Z]^{-1}\E[ZI\{D=t\}]$. 
In particular, if the generic pair $(D,Z)$ has independent components or if $Z\equiv 1$, then 
$(\widetilde{q}_1,\dots,\widetilde{q}_T)=(q_1,\dots,q_T)$.  
\end{remark}

\subsection{Conditional mean squared error of prediction}

The natural measure of prediction error is 
\begin{align}\label{eq:msep}
\E\big[(C^{\alpha}_{i,T}-\widehat{C}^{\alpha}_{i,T})^2 \mid \mathcal{D}^{\alpha}\big], 
\end{align}
where $\mathcal{D}^{\alpha}$ is the $\sigma$-algebra generated by $\{C^{\alpha}_{j,t}:j,t\in\mathcal{T},j+t\leq T+1\}$, the run-off triangle that is fully observed at the time of prediction. Since we are considering large exposure limits, the conditional expectation \eqref{eq:msep} diverges as $\alpha\to\infty$ and is hence not meaningful. However, we show (Theorems \ref{thm:msep_second_term}, \ref{thm:msep_first_term} and \ref{thm:msep_third_term} together with Remark \ref{rem:msep_joint_convergence}) that there exists a random variable $L$ such that the standardized mean squared error of prediction converges in distribution,  
\begin{align}\label{eq:convtoL}
\E\bigg[\frac{(C^{\alpha}_{i,T}-\widehat{C}^{\alpha}_{i,T})^2}{C^{\alpha}_{i,T-i+1}} \mid \mathcal{D}^{\alpha}\bigg]
\distr L,  
\end{align}
and that the limit $L$ has a natural $\mathcal{D}^{\alpha}$-measurable estimator $\widehat{L}^{\alpha}$ (Remarks \ref{rem:msep_second_term},  \ref{rem:msep_first_term} and  \ref{rem:msep_third_term}). Consequently, the natural estimator of the prediction error \eqref{eq:msep} is $C^{\alpha}_{i,T-i+1}\widehat{L}^{\alpha}$:
\begin{align*}
\E\big[(C^{\alpha}_{i,T}-\widehat{C}^{\alpha}_{i,T})^2 \mid \mathcal{D}^{\alpha}\big]
=C^{\alpha}_{i,T-i+1} \E\bigg[\frac{(C^{\alpha}_{i,T}-\widehat{C}^{\alpha}_{i,T})^2}{C^{\alpha}_{i,T-i+1}} \mid \mathcal{D}^{\alpha}\bigg]
\approx C^{\alpha}_{i,T-i+1}\widehat{L}^{\alpha}.
\end{align*}
Our aim is to arrive at an estimator of conditional mean squared error of prediction that coincides with Mack's estimator \eqref{eq:mack_msep_estimator}, and this is not in general true in the setting of the general model. Therefore, we need to consider the special model (SM1)-(SM3). 

Combining Theorems \ref{thm:msep_second_term}, \ref{thm:msep_first_term}, \ref{thm:msep_third_term} and Remarks \ref{rem:msep_second_term}, \ref{rem:msep_first_term}, \ref{rem:msep_third_term} below we show that  
\begin{align}\label{eq:msep_estimator} 
C^{\alpha}_{i,T-i+1}\widehat{L}^{\alpha}=(\widehat{C}^{\alpha}_{i,T-i+1})^2\sum_{t=T-i+1}^{T-1}\frac{\widehat{\sigma}^2_t}{\widehat{f}^2_t}
\bigg(\frac{1}{\widehat{C}^{\alpha}_{i,t}}+\frac{1}{\sum_{j=1}^{T-t}C^{\alpha}_{j,t}}\bigg)
\end{align}
which coincides with the estimator of conditional mean squared error of prediction obtained by Mack in \cite{Mack-93}. 
Note that in \eqref{eq:msep_estimator} we use the notation 
\begin{align*}
\widehat{C}^{\alpha}_{i,T-i+1}=C^{\alpha}_{i,T-i+1}, \quad 
\widehat{C}^{\alpha}_{i,t}=C^{\alpha}_{i,T-i+1}\prod_{s=T-i+1}^{t-1}\widehat{f}_s, \quad t>T-i+1.
\end{align*}
Note that $C^{\alpha}_{i,T-i+1}$ is independent of $\widehat{f}_{T-i+1},\dots,\widehat{f}_{T-1}$ since the latter estimators are functions of only data from accident years $\leq i-1$. Hence, $\widehat{C}^{\alpha}_{i,T}=C^{\alpha}_{i,T-i+1}\prod_{s=T-i+1}^{T-1}\widehat{f}_s$ is a product of two independent factors. 
In order to verify the convergence in \eqref{eq:convtoL}, note that the left-hand side in \eqref{eq:convtoL} can be expressed as
\begin{align}
&\E\bigg[\frac{(C^{\alpha}_{i,T}-C^{\alpha}_{i,T-i+1}\prod_{s=T-i+1}^{T-1}f_s)^2}{C^{\alpha}_{i,T-i+1}} \mid \mathcal{D}^{\alpha}\bigg] \label{eq:msep_first_term}\\
&\quad+C^{\alpha}_{i,T-i+1}\bigg(\prod_{s=T-i+1}^{T-1}f_s-\prod_{s=T-i+1}^{T-1}\widehat{f}_s\bigg)^2 \label{eq:msep_second_term} \\
&\quad+2\E\bigg[C^{\alpha}_{i,T}-C^{\alpha}_{i,T-i+1}\prod_{s=T-i+1}^{T-1}f_s\mid \mathcal{D}^{\alpha}\bigg]\bigg(\prod_{s=T-i+1}^{T-1}f_s-\prod_{s=T-i+1}^{T-1}\widehat{f}_s\bigg) \label{eq:msep_third_term}. 
\end{align}
In the literature, the first term \eqref{eq:msep_first_term} (upon multiplication by $C^{\alpha}_{i,T-i+1}$) is referred to as process variance, and the second term \eqref{eq:msep_second_term} (upon multiplication by $C^{\alpha}_{i,T-i+1}$) is referred to as estimation error. In the setting of the distribution-free chain ladder, \eqref{eq:msep_first_term} is a conditional variance. However, in our setting (the general or special model, see Section \ref{sec:themodel}) this term is not a conditional variance. Hence, we will not use the terminology ``process variance''.  
Note that the two factors in \eqref{eq:msep_third_term} are independent because of independent accident years. This fact will enable us to study the asymptotic behavior of \eqref{eq:msep_third_term}, convergence in distribution, and verify that the limit distribution has zero mean.  

Theorem \ref{thm:msep_second_term} below shows that the second term \eqref{eq:msep_second_term} converges in distribution in the setting of the general model. 
Theorem \ref{thm:msep_first_term} below shows that the first term \eqref{eq:msep_first_term} converges in distribution in the setting of the special model. In fact, the Poisson assumption for the counting variables is not not needed for convergence in distribution. However, we need it in order to obtain an estimator of conditional mean squared error of prediction that coincides with the estimator derived by Mack in \cite{Mack-93}. 
Theorem \ref{thm:msep_third_term} below shows that the third term \eqref{eq:msep_third_term} converges in distribution in the setting of the special model. 
Remark \ref{rem:msep_joint_convergence} below clarifies that the sum of the terms converges in distribution in the setting of the special model. 

\begin{theorem}\label{thm:msep_second_term}
Consider the general model (GM1)-(GM4). 
For each $i\in\mathcal{T}$ with $i\geq 2$, there exists $\gamma_i\in(0,\infty)$ such that  
\begin{align*}
C^{\alpha}_{i,T-i+1}\bigg(\prod_{s=T-i+1}^{T-1}f_s-\prod_{s=T-i+1}^{T-1}\widehat{f}_s\bigg)^2 
\distr \gamma_i^2\chi^2_1.
\end{align*}
If $Z$ and $D$ are independent, then 
\begin{align*}
\gamma_i^2=\lambda_i\E[ZI\{D\leq T-i+1\}]\prod_{s=T-i+1}^{T-1}f_s^2\sum_{t=T-i+1}^{T-1}\frac{\sigma^2_t/f_t^2}{\sum_{j=1}^{T-t}\lambda_j \E[ZI\{D\leq t\}]}.
\end{align*}
\end{theorem}

\begin{remark}\label{rem:msep_second_term}
Motivated by \eqref{eq:f_hat} and \eqref{eq:sigma_hat} we estimate $f_t$ by $\widehat{f}_t$ and $\sigma^2_t$ by $\widehat{\sigma}_t^2$. Since $\alpha^{-1}C^{\alpha}_{j,t}\as \lambda_j\E[ZI\{D\leq t\}]$ we estimate $\lambda_j\E[ZI\{D\leq t\}]$ by $\alpha^{-1}C^{\alpha}_{j,t}$. Hence, the estimator of $\gamma_i^2$ is 
\begin{align*}
\widehat{\gamma}_i^2=C^{\alpha}_{i,T-i+1}\prod_{s=T-i+1}^{T-1}\widehat{f}_s^2\sum_{t=T-i+1}^{T-1}\frac{\widehat{\sigma}^2_t/\widehat{f}_t^2}{\sum_{j=1}^{T-t}C^{\alpha}_{j,t}}.
\end{align*} 
Consequently, the estimator of 
\begin{align*}
\big(C^{\alpha}_{i,T-i+1}\big)^2\bigg(\prod_{s=T-i+1}^{T-1}f_s-\prod_{s=T-i+1}^{T-1}\widehat{f}_s\bigg)^2
\end{align*}
is $C^{\alpha}_{i,T-i+1}\widehat{\gamma}_i^2$ which equals 
\begin{align*}
(C^{\alpha}_{i,T-i+1})^2\prod_{s=T-i+1}^{T-1}\widehat{f}_s^2\sum_{t=T-i+1}^{T-1}\frac{\widehat{\sigma}^2_t/\widehat{f}_t^2}{\sum_{j=1}^{T-t}C^{\alpha}_{j,t}}.
\end{align*} 
and coincides with Mack's estimator (see \cite{Mack-93}, p.~219). 
\end{remark}
 
\begin{theorem}\label{thm:msep_first_term}
Consider the special model (SM1)-(SM3). 
For each $i\in\mathcal{T}$ with $i\geq 2$,  
\begin{align}
&\E\bigg[\frac{(C^{\alpha}_{i,T}-C^{\alpha}_{i,T-i+1}\prod_{s=T-i+1}^{T-1}f_s)^2}{C^{\alpha}_{T-i+1}} \mid \mathcal{D}^{\alpha}\bigg] \nonumber \\
&\quad=\frac{\alpha}{C^{\alpha}_{T-i+1}}\bigg(\frac{\E[M^{\alpha}_i]}{\alpha}\E[Z^2]\P(D>T-i+1)+(H^{\alpha})^2\bigg(\prod_{s=T-i+1}^{T-1}f_s-1\bigg)^2\bigg) \nonumber \\
&\quad\distr 
\frac{\E[Z^2]}{\E[Z]}\bigg(\bigg(\prod_{s=T-i+1}^{T-1}f_s-1\bigg)+\bigg(\prod_{s=T-i+1}^{T-1}f_s-1\bigg)^2\chi^2_1\bigg), \label{eq:msep_first_term_lim}  
\end{align}
where 
\begin{align*}
(H^{\alpha})^2=\frac{(C^{\alpha}_{i,T-i+1}-\E[C^{\alpha}_{i,T-i+1}])^2}{\alpha}
\distr \lambda_i\E[Z^2]\P(D\leq T-i+1)\chi^2_1=H^2.
\end{align*}
In particular, the expectation of the limit variable in \eqref{eq:msep_first_term_lim} is 
\begin{align}
&\frac{\lambda_i\E[Z^2]\P(D>T-i+1)+\E[H^2](\prod_{s=T-i+1}^{T-1}f_s-1)^2}{\lambda_i\E[Z]\P(D\leq T-i+1)} \nonumber \\
&\quad=\sum_{t=T-i+1}^{T-1}f_{T-i+1}\dots f_{t-1}\sigma^2_t f_{t+1}^2\dots f_{T-1}^2. \label{eq:msep_first_term_exp}
\end{align}
\end{theorem}

\begin{remark}\label{rem:msep_first_term}
Since \eqref{eq:msep_first_term_exp} equals 
\begin{align*}
C^{\alpha}_{i.T-i+1}\prod_{s=T-i+1}^{T-1}f_s^2\sum_{t=T-i+1}^{T-1}\frac{\sigma^2_t/f_t^2}{C^{\alpha}_{i,T-i+1}\prod_{s=T-i+1}^{t-1}f_s}, 
\end{align*}
estimating $f_t$ by $\widehat{f}_t$ and $\sigma^2_t$ by $\widehat{\sigma}^2_t$ gives the estimator of \eqref{eq:msep_first_term_exp} given by 
\begin{align*}
C^{\alpha}_{i.T-i+1}\prod_{s=T-i+1}^{T-1}\widehat{f}_s^2\sum_{t=T-i+1}^{T-1}\frac{\widehat{\sigma}^2_t/\widehat{f}_t^2}{\widehat{C}^{\alpha}_{i,t}}. 
\end{align*}
Consequently, we estimate 
\begin{align*}
C^{\alpha}_{i,T-i+1}\E\bigg[\frac{(C^{\alpha}_{i,T}-C^{\alpha}_{i,T-i+1}\prod_{s=T-i+1}^{T-1}f_s)^2}{C^{\alpha}_{T-i+1}} \mid \mathcal{D}^{\alpha}\bigg] 
\end{align*}
by
\begin{align*}
(C^{\alpha}_{i.T-i+1})^2\prod_{s=T-i+1}^{T-1}\widehat{f}_s^2\sum_{t=T-i+1}^{T-1}\frac{\widehat{\sigma}^2_t/\widehat{f}_t^2}{\widehat{C}^{\alpha}_{i,t}}
\end{align*} 
which coincides with Mack's estimator (see \cite{Mack-93}, p.~218).
\end{remark}

\begin{remark}
Convergence of the conditional expectations considered in Theorem \ref{thm:msep_first_term} does not require the Poisson assumption for the counting variables. However, we have used the fact that $\E[M^{\alpha}_i]=\var(M^{\alpha}_i)$ to derive the limit in \eqref{eq:msep_first_term_lim}. If $\E[M^{\alpha}_i]$ and $\var(M^{\alpha}_i)$ would increase with $\alpha$ at rates that differ asymptotically, then  a limit corresponding to \eqref{eq:msep_first_term_lim} would look differently and consequently we would arrive at an estimator of conditional mean squared error of prediction that would differ from the one obtained by Mack in \cite{Mack-93}.  
\end{remark}

\begin{theorem}\label{thm:msep_third_term}
Consider the special model (SM1)-(SM3). 
Let 
\begin{align*}
A^{\alpha}_1&=\alpha^{-1/2}\E\bigg[C^{\alpha}_{i,T}-C^{\alpha}_{i,T-i+1}\prod_{s=T-i+1}^{T-1}f_s\mid \mathcal{D}^{\alpha}\bigg], \\
A^{\alpha}_2&=\alpha^{1/2}\bigg(\prod_{s=T-i+1}^{T-1}f_s-\prod_{s=T-i+1}^{T-1}\widehat{f}_s\bigg).
\end{align*}
Then $(A^{\alpha}_1)_{\alpha\geq 0}$ and $(A^{\alpha}_2)_{\alpha\geq 0}$ are independent and both converges in distribution to normally distributed random variables with zero means. In particular, $(A^{\alpha}_1A^{\alpha}_2)_{\alpha\geq 0}$ converges in distribution to a random variable with zero mean. 
\end{theorem}

\begin{remark}\label{rem:msep_third_term}
By Theorem \ref{thm:msep_third_term} the third term \eqref{eq:msep_third_term} in the expression for the standardized mean squared error of prediction converges in distribution to a random variable with zero mean. Consequently, we estimate \eqref{eq:msep_third_term} by $0$. 
\end{remark}

\begin{theorem}\label{thm:renewal_clt}
Suppose that for each accident year $j$, $(M^{\alpha}_j)_{\alpha\geq 0}$ is a renewal counting process given by $M^{\alpha}_j=\sup\{m\geq 1:T_{j,m}\leq \alpha\}$, where the steps $Y_{j,k}$ of the random walk $T_{j,m}=\sum_{k=1}^m Y_{j,k}$ satisfies $\E[Y_{j,k}]=1/\lambda_j$ and $\var(Y_{j,k})<\infty$. Then  
\begin{align*}
S_j^{\alpha}=\sum_{k=1}^{M^{\alpha}_j}Z_{j,k}\big(I\{D_{j,k}=1\},\dots,I\{D_{j,k}=T\}\big)
\end{align*}
satisfies $\alpha^{-1/2}(S_j^{\alpha}-\E[S_j^{\alpha}])\distr \Normal_T(0,\Sigma)$, where 
\begin{align*}
\Sigma_{s,t}&=\lambda_j\E[Z^2I\{D=s\}I\{D=t\}]\\
&\quad+\lambda_j(\lambda_j^2\var(Y)-1)\E[ZI\{D=s\}]\E[ZI\{D=t\}].
\end{align*}
\end{theorem}

\begin{corollary}\label{cor:renewal_clt}
Consider the setting of Theorem \ref{thm:renewal_clt}. Let 
\begin{align*}
H^{\alpha}&=\alpha^{-1/2}\big(C^{\alpha}_{i,T-i+1}-\E[C^{\alpha}_{i,T-i+1}]\big), \\
F^{\alpha}&=\alpha^{-1/2}\big(C^{\alpha}_{i,T}-C^{\alpha}_{i,T-i+1}-\E[C^{\alpha}_{i,T}-C^{\alpha}_{i,T-i+1}]\big). 
\end{align*}
Then $(H^{\alpha},F^{\alpha})\distr (H,F)$, where $(H,F)$ is jointly normally distributed with 
\begin{align*}
\var(H)&=\lambda_j \E[Z^2I\{D\leq t\}]+\lambda_j(\lambda_j^2\var(Y)-1)\E[ZI\{D\leq t\}]^2,\\
\var(F)&=\lambda_j \E[Z^2I\{D> t\}]+\lambda_j(\lambda_j^2\var(Y)-1)\E[ZI\{D> t\}]^2,\\
\cov(H,F)&=\lambda_j(\lambda_j^2\var(Y)-1)\E[ZI\{D\leq t\}]\E[ZI\{D> t\}].
\end{align*}
\end{corollary}

\begin{remark}\label{rem:renewal_clt}
If $(M^{\alpha}_j)_{\alpha\geq 0}$ is a homogeneous Poisson process, then $\var(Y)=\lambda_j^{-2}$, the random vectors $S_j^{\alpha}$ in Theorem \ref{thm:renewal_clt} have independent components, and $H^{\alpha}$ and $F^{\alpha}$ in Corollary \ref{cor:renewal_clt} are independent. 
\end{remark}

\begin{remark}\label{rem:msep_joint_convergence}
Theorems \ref{thm:msep_second_term}, \ref{thm:msep_first_term} and \ref{thm:msep_third_term} show convergence in distribution separately for the three terms \eqref{eq:msep_first_term}, \eqref{eq:msep_second_term} and \eqref{eq:msep_third_term} of conditional mean squared error of prediction. We treat them separately since we want to emphasize that convergence to the appropriate limits occurs under different assumptions; only for two of the terms we use the compound Poisson assumption of the special model. However, the sum of the  terms converges in distribution under the assumptions made in Theorem \ref{thm:msep_first_term}. This convergence of the sum is a consequence of the convergence in distribution of the random vectors 
$\alpha^{-1/2}(S_j^{\alpha}-\E[S_j^{\alpha}])$ in Theorem \ref{thm:renewal_clt}. 
That the convergence in distribution in Theorems \ref{thm:msep_second_term}, \ref{thm:msep_first_term} and \ref{thm:msep_third_term} can be extended to joint convergence in distribution can then be verified by combining the convergence of $\alpha^{-1/2}(S_j^{\alpha}-\E[S_j^{\alpha}])$ in Theorem  \ref{thm:renewal_clt} with an application of the continuous mapping theorem for weak convergence together with Slutsky's theorem. Such an argument verifies that
\begin{align*}
\E\bigg[\frac{(C^{\alpha}_{i,T}-C^{\alpha}_{i,T-i+1}\prod_{s=T-i+1}^{T-1}\widehat{f}_s)^2}{C^{\alpha}_{T-i+1}} \mid \mathcal{D}^{\alpha}\bigg] 
\distr L=L^{(1)}+L^{(2)}+L^{(3)},
\end{align*} 
where $L^{(1)}$, $L^{(2)}$ and $L^{(3)}$ correspond to the limits in Theorems \ref{thm:msep_second_term}, \ref{thm:msep_first_term} and \ref{thm:msep_third_term}. 
\end{remark}

\section{Numerical illustration}\label{sec:numeric_example}

In the setting of the special model, we may simulate a run-off triangle $\{C^{\alpha}_{j,t}:j,t\in\mathcal{T},j+t\leq T+1\}$ and explicitly compute the standardized conditional mean squared error of prediction (standardized means division by $C^{\alpha}_{T-i+1}$) in \eqref{eq:convtoL} as a known function of the simulated run-off triangle. For the same run-off triangle, we may compute the standardized estimator of mean squared error by Mack,
\begin{align}\label{eq:MackMSE}
\widehat{L}^{\alpha}=\frac{(\widehat{C}^{\alpha}_{i,T})^2}{C^{\alpha}_{i,T-i+1}}
\sum_{t=T-i+1}^{T-1}\frac{\widehat{\sigma}^2_t}{\widehat{f}^2_t}
\bigg(\frac{1}{\widehat{C}^{\alpha}_{i,t}}+\frac{1}{\sum_{j=1}^{T-t}C^{\alpha}_{j,t}}\bigg),
\end{align}
and then compare the two random variables, or their distributions. 

We first show how to explicitly compute the standardized conditional mean squared error of prediction. 
Since $C^{\alpha}_{i,T}=C^{\alpha}_{i,T-i+1}+\sum_{k=1}^{N^{\alpha}}Z_k$ with $N^{\alpha}\sim\Pois(\alpha\lambda_i\sum_{t=T-i+2}^{T}q_t)$ independent of the i.i.d. sequence $(Z_k)$, and 
\begin{align*}
\E\bigg[\sum_{k=1}^{N^{\alpha}}Z_k\bigg]&=\E[N^{\alpha}]\E[Z], \\
\E\bigg[\bigg(\sum_{k=1}^{N^{\alpha}}Z_{k}\bigg)^2\bigg]&=\E[N^{\alpha}]\var(Z)+\E[N^{\alpha}]\E[Z]^2+\E[N^{\alpha}]^2\E[Z]^2,
\end{align*}
we may use the independence between $\sum_{k=1}^{N^{\alpha}}Z_k$ and $\mathcal{D}^{\alpha}$ to get 
\begin{align}
L^{\alpha}&=\E\bigg[\frac{(C^{\alpha}_{i,T}-C^{\alpha}_{i,T-i+1}\prod_{s=T-i+1}^{T-1}\widehat{f}_s)^2}{C^{\alpha}_{T-i+1}} \mid \mathcal{D}^{\alpha}\bigg] \label{eq:RealPoissonMSE}\\
&=(C^{\alpha}_{T-i+1})^{-1}\E\bigg[\bigg(\sum_{k=1}^{N^{\alpha}}Z_{k}\bigg)^2\bigg]-2\bigg(\prod_{s=T-i+1}^{T-1}\widehat{f}_s-1\bigg)\E\bigg[\sum_{k=1}^{N^{\alpha}}Z_k\bigg] \nonumber \\
&\quad+C^{\alpha}_{i,T-i+1}\bigg(\prod_{s=T-i+1}^{T-1}\widehat{f}_s-1\bigg)^2. \nonumber
\end{align}

From Theorems \ref{thm:msep_second_term}, \ref{thm:msep_first_term} and \ref{thm:msep_third_term} together with Remark \ref{rem:msep_joint_convergence} we know that $L^{\alpha}\distr L$ and we may compute $\E[L]$ explicitly. We have not shown convergence in distribution for $\widehat{L}^{\alpha}$ but it follows from Theorem \ref{thm:estimation} and Slutsky's theorem that each term in the expression for $\widehat{L}^{\alpha}$ converges in distribution, and the corresponding expectations of the limits add up to $\E[L]$. Hence, if we draw many realizations of run-off triangles based on the special model, and convert these into a random sample from the distribution of $L^{\alpha}-\widehat{L}^{\alpha}$, then we expect the empirical mean to be approximately zero.  

For the numerical illustration, we take the claims data from Table 1 in Mack \cite{Mack-93}, originally presented by Taylor and Ashe \cite{Taylor-Ashe-83}, in order to choose values for the model parameters of exposure and distribution of delay. Applying the formula from Remark \ref{rem:f_to_q}, we can transform the development factors $\widehat{f}_t$ corresponding to Table 1 in \cite{Mack-93} into
$$
(\widehat{q}_{t})_{t=1}^{T} = (0.069, 0.172, 0.180, 0.194,  0.107, 0.075, 0.069, 0.047, 0.070, 0.018).
$$
For the exposures, we simply use the first column of the run-off triangle in Mack (1993) and normalize it by dividing by its first entry (this procedure suffices for illustration, more sophisticated estimation could be considered). This yields
$$
(\widehat{\lambda}_i)_{i=1}^T = (1.000, 0.984, 0.812, 0.868, 1.239, 1.107, 1.230, 1.005, 1.053, 0.961)
$$
across accident years. 
For simplicity, we choose $Z \equiv 1$ and $\alpha = 4,000,000$, which roughly corresponds to the order of magnitude as can be found in  \cite{Mack-93}. 
We generate $100,000$ realizations of run-off triangles and for each one compute both the true standardized conditional mean squared error \eqref{eq:RealPoissonMSE}, as well as the standardized version of Mack's estimator of conditional mean squared error \eqref{eq:MackMSE} for accident years $i=3, 5$ and $8$. The results can be seen in Figure \ref{fig:hist}. The results are not sensitive to the value chosen for $\alpha$, the histograms in Figure \ref{fig:hist} are essentially indistinguishable from those with $\alpha=10,000$.

\begin{figure}[htp]
\centering
\includegraphics[width=.333\textwidth]{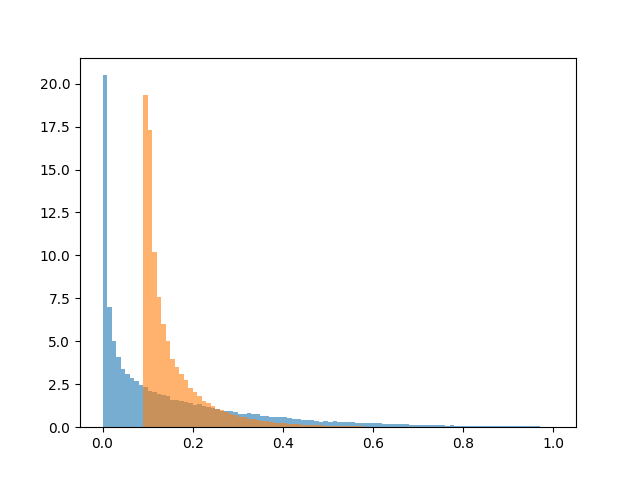}\hfill
\includegraphics[width=.333\textwidth]{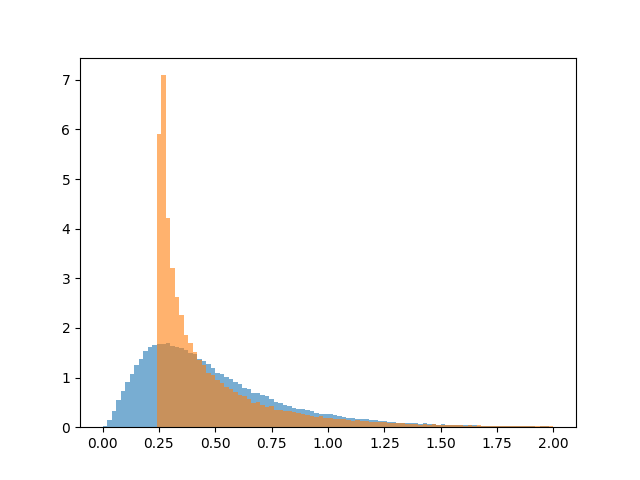}\hfill
\includegraphics[width=.333\textwidth]{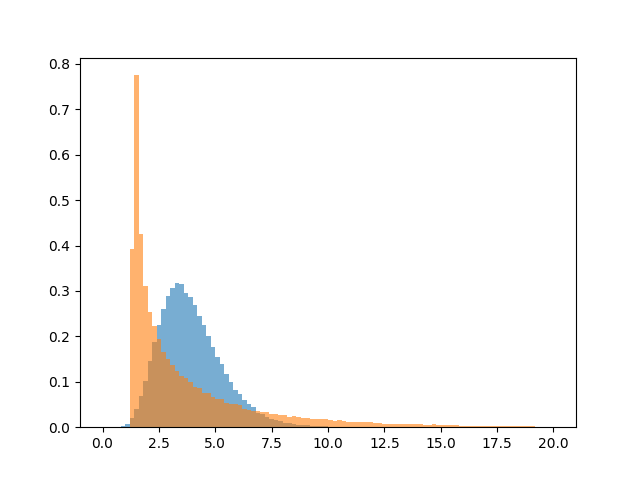}
\caption{Blue histograms: standardized Mack's estimator \eqref{eq:MackMSE} of conditional mean squared error. Orange histograms: true standardized conditional mean squared error \eqref{eq:RealPoissonMSE}. The three plots shown correspond to accident years $i=3,5,8$ from left to right.}
\label{fig:hist}
\end{figure}

\section{proofs}\label{sec:proofs}

Before the proof of Theorem \ref{thm:estimation} we state a result, on stochastic representations of norms of multivariate normal random vectors, that will be used in the proof of Theorem \ref{thm:estimation}. 

\begin{lemma}\label{lem:normofgaussian}
If $W\sim \Normal_n(0,\Sigma)$, then $W^{\trans}W\eqdis \sum_{i=1}^n \mu_iQ^2_i$, where $Q_1,\dots,Q_n$ are independent and standard normal and $\mu_1,\dots,\mu_n$ are the eigenvalues of $\Sigma$. 
\end{lemma}

\begin{proof}[Proof of Lemma \ref{lem:normofgaussian}]
Write $\Sigma=LL^{\trans}$ and note that $W\eqdis LQ$ with $Q\sim \Normal_n(0,I)$. Hence, $W^{\trans}W\eqdis Q^{\trans}L^{\trans}LQ$. The matrix $L^{\trans}L$ is orthogonally diagonizable and has the same eigenvalues as $\Sigma=LL^{\trans}$. Write $L^{\trans}L=O^{\trans}D O$, where $O$ is orthogonal and $D=\diag(\mu_1,\dots,\mu_n)$. Hence, 
\begin{align*}
W^{\trans}W\eqdis Q^{\trans}L^{\trans}LQ\eqdis Q^{\trans}O^{\trans}DOQ\eqdis Q^{\trans}DQ
\end{align*}
since $OQ\eqdis Q$. 
\end{proof}

\begin{proof}[Proof of Theorem \ref{thm:estimation}]
We first prove \eqref{eq:f_hat}. 
Note that, for $1\leq i_0<i_1\leq T$, using Theorem 2.1 in \cite{Gut-2009}, 
\begin{align*}
\frac{1}{\alpha}\sum_{i=i_0}^{i_1}C^{\alpha}_{i,t+1}
&=\sum_{i=i_0}^{i_1}\frac{M^{\alpha}_i}{\alpha}\frac{1}{M^{\alpha}_i}\sum_{k=1}^{M^{\alpha}_i}Z_{i,k}I\{D_{i,k}\leq t+1\} \\
&\as \E[ZI\{D\leq t+1\}]\sum_{i=i_0}^{i_1}\lambda_i. 
\end{align*}
Consequently, 
\begin{align*}
\frac{\sum_{i=i_0}^{i_1}C^{\alpha}_{i,t+1}}{\sum_{i=i_0}^{i_1}C^{\alpha}_{i,t}} \as \frac{\E[ZI\{D\leq t+1\}]}{\E[ZI\{D\leq t\}]}. 
\end{align*}
In order to prove \eqref{eq:clpredictor}, Note that, similarly to the above,  
\begin{align*}
\frac{C^{\alpha}_{i,T}}{C^{\alpha}_{i,t}}\as \frac{\E[ZI\{D\leq T\}]}{\E[ZI\{D\leq t\}]}
\quad\text{and}\quad 
\prod_{s=t}^{T-1}\widehat{f}_s \as \frac{\E[ZI\{D\leq T\}]}{\E[ZI\{D\leq t\}]}.
\end{align*}
We proceed to the more involved task of proving \eqref{eq:sigma_hat}. 
For $j=i_0,\dots,i_1$, let 
\begin{align*}
 W^{\alpha}_j=\frac{
 \alpha^{-1/2} \Big((C^{\alpha}_{j,t+1} - C^{\alpha}_{j,t})  - \frac{\sum_{i=i_0}^{i_1} (C^{\alpha}_{i,t+1} - C^{\alpha}_{i,t})}{\sum_{i=i_0}^{i_1}C^{\alpha}_{i,t}} C^{\alpha}_{j,t} \Big)
}{
(\alpha^{-1}C^{\alpha}_{j,t})^{1/2}
}.
\end{align*}
Some algebra shows that 
\begin{align*}
\big(W^{\alpha}_j\big)^2=C^{\alpha}_{i,t}\bigg(\frac{C^{\alpha}_{i,t+1}}{C^{\alpha}_{i,t}}-\widehat{f}_{t}\bigg)^2
\end{align*}
i.e.~the $j$th term in the sum in the expression for $\widehat{\sigma}^2_t$. 
The numerator of $W_j^\alpha$ can be written as 
\begin{align*}
\begin{split}
&\left( 1 - \frac{
C^{\alpha}_{j,t}
}{
\sum_{i=i_0}^{i_1} C^{\alpha}_{i,t}
}\right)
\alpha^{-1/2}\sum_{k=1}^{M_j^\alpha} Z_{j,k}\left( I\{D_{j,k} = t+1\} - \frac{\E[ZI\{D= t+1\}]}{\E[ZI\{D\leq t\}]} I\{D_{j,k} \leq t\} \right) \\
 &-\frac{
C^{\alpha}_{j,t}
}{
\sum_{i=i_0}^{i_1} C^{\alpha}_{i,t}
}
\alpha^{-1/2} \sum_{i=i_0, i \neq j}^{i_1} \sum_{k=1}^{M_i^\alpha} Z_{i,k}\left( I\{D_{i,k} = t+1\} - \frac{\E[ZI\{D= t+1\}]}{\E[ZI\{D\leq t\}]} I\{D_{i,k} \leq t\} \right) . 
\end{split}
\end{align*}
We can now write $W^{\alpha}=B^{\alpha}U^{\alpha}$, where 
\begin{align*}
U_{j}^\alpha = \alpha^{-1/2}\sum_{k=1}^{M_j^\alpha} Z_{k, j}\bigg( I\{D_{j,k} = t+1\} - \frac{\E[ZI\{D= t+1\}]}{\E[ZI\{D\leq t\}]} I\{D_{j,k} \leq t\} \bigg) 
\end{align*}
and $B^\alpha$ is a square matrix with entries 
$$
B_{j, l}^\alpha = 
\begin{cases}
\left( \alpha^{-1} C^{\alpha}_{j,t} \right)^{-1/2}\left( 1 - \frac{
C^{\alpha}_{j,t}
}{
\sum_{i=i_0}^{i_1} C^{\alpha}_{i,t}
}\right), &j = l \\
\left( \alpha^{-1} C^{\alpha}_{j,t} \right)^{-1/2}\left( - \frac{
C^{\alpha}_{j,t}
}{
\sum_{i=i_0}^{i_1} C^{\alpha}_{i,t}
}\right), &j \neq l.
\end{cases}
$$
The multivariate Central Limit Theorem together with Theorem 1.1 in \cite{Gut-2009} yield $U^\alpha \distr U$, where 
$U \sim \Normal_{i_1 - i_0 + 1}(0,c_t^2 \diag(\lambda_{i_0}, \dots, \lambda_{i_1}))$ 
with 
\begin{align*}
c_t^2 &= \var\bigg( Z\bigg( I\{D = t+1\} - \frac{\E[ZI\{D= t+1\}]}{\E[ZI\{D\leq t\}]} I\{D \leq t\} \bigg) \bigg)\\
&=\E[Z^2I\{D = t+1\}]+\bigg(\frac{\E[ZI\{D= t+1\}]}{\E[ZI\{D\leq t\}]}\bigg)^2\E[Z^2I\{D \leq t\}].
\end{align*}
By the strong law of large numbers, $B^\alpha \as B$, where 
\begin{align*}
B_{j, l} = & \E[ZI\{D\leq t\}]^{-1/2} \cdot
\begin{cases} 
\lambda_j^{-1/2} \left( 1 - 
\frac{\lambda_j}{ \sum_{i=i_0}^{i_1} \lambda_i} \right), &j = l\\
\lambda_j^{-1/2}\left(-
\frac{\lambda_j}{ \sum_{i=i_0}^{i_1} \lambda_i} \right), &j \neq l.
\end{cases}
\end{align*}
Hence, by Slutsky's theorem (multivariate version), $W^\alpha = B^\alpha U^\alpha \distr BU = W$, where $W\sim \Normal_{i_1-i_0+1}(0,\Sigma)$ with 
$$
\Sigma =  B \cov(U) B^{\trans} = \frac{c_t^2}{\E[ZI\{D\leq t\}]} \widetilde \Sigma
=\sigma_t^2 \widetilde \Sigma
$$
The eigenvalues of $\widetilde{\Sigma}$ are $\mu_1 = 1, \mu_2 = 0$ with corresponding eigenspaces
$$
\text{Eig}_1 = \text{span}\left(
\begin{bmatrix}
\lambda_{i_0 + 1}^{1/2} \\
- \lambda_{i_0}^{1/2}\\
0 \\
0 \\
\vdots \\
0
\end{bmatrix},
\begin{bmatrix}
\lambda_{i_0 + 2}^{1/2} \\
0 \\
-\lambda_{i_0}^{1/2} \\
0 \\
\vdots \\
0
\end{bmatrix}, \dots,
\begin{bmatrix}
\lambda_{i_1}^{1/2} \\
0 \\
0\\
\vdots \\
0 \\
- \lambda_{i_0}^{1/2}
\end{bmatrix}
\right),
 \quad
 \text{Eig}_0 = \text{span}\left(
\begin{bmatrix}
\lambda_{i_0}^{1/2} \\
\vdots \\
\lambda_{i_1}^{1/2}
\end{bmatrix}
\right)
$$
and hence geometric multiplicities $i_1 - i_0$ and $1$, respectively. By Lemma \ref{lem:normofgaussian},  
\begin{align*}
\sum_{i=i_0}^{i_1} W_i^2 \eqdis \sigma_t^2\sum_{i=i_0}^{i_1-1} Q_i^2,
\end{align*}
where $Q_{i_0}, \dots, Q_{i_1 -1}$ are independent and standard normal. Altogether, we have shown that 
\begin{align*}
\widehat{\sigma}^2_{t}
\eqdis \frac{1}{i_1-i_0}\sum_{i=i_0}^{i_1}\big(W^{\alpha}_i\big)^2 
\distr \sigma_t^2\frac{ \chi^2_{i_1 - i_0}}{i_{1} - i_{0}}.
\end{align*}
\end{proof}

\begin{proof}[Proof of Theorem \ref{thm:msep_second_term}]
Write $S_t=\widehat{f}_{T-i+1}\dots \widehat{f}_{t-1}(f_t-\widehat{f}_t)f_{t+1}\dots f_{T-1}$ and note, as noted in \cite{Mack-93}, that 
$f_{T-i+1}\dots f_{T-1}-\widehat{f}_{T-i+1}\dots\widehat{f}_{T-1}=\sum_{t=T-i+1}^{T-1}S_t$. Hence, the statement of the theorem follows if we show the appropriate convergence in distribution of   
\begin{align}\label{eq:msep_second_term_rep}
\frac{\alpha^{-1/2}C^{\alpha}_{i,T-i+1}}{(\alpha^{-1}C^{\alpha}_{i,T-i+1})^{1/2}}\sum_{t=T-i+1}^{T-1}S_t.
\end{align}
Write 
\begin{align*}
S_t&=\bigg(\prod_{s=T-i+1}^{t-1}\widehat{f}_s\bigg)\bigg(\prod_{s=t+1}^{T-1}f_s\bigg)\bigg(\frac{\E[ZI\{D\leq t+1\}]}{\E[ZI\{D\leq t\}]}-\frac{\sum_{j=1}^{T-t}C^{\alpha}_{j,t+1}}{\sum_{j=1}^{T-t}C^{\alpha}_{j,t}}\bigg)\\
&=\frac{(\prod_{s=T-i+1}^{t-1}\widehat{f}_s)(\prod_{s=t+1}^{T-1}f_s)}{\sum_{j=1}^{T-t}C^{\alpha}_{j,t}}
\alpha^{1/2}\sum_{j=1}^{T-t}\bigg(\frac{M^{\alpha}_j}{\alpha}\bigg)^{1/2}U^{\alpha}_{j,t}, 
\end{align*}
where 
\begin{align*}
U^{\alpha}_{j,t}=(M^{\alpha}_j)^{-1/2}\sum_{k=1}^{M^{\alpha}_j}
Z_{j,k}\bigg(\frac{\E[ZI\{D=t+1\}]}{\E[ZI\{D\leq t\}]}I\{D_{j,k}\leq t\}-I\{D_{j,k}=t+1\}\bigg). 
\end{align*}
Therefore, we may write \eqref{eq:msep_second_term_rep} as 
\begin{align*}
\sum_{t=T-i+1}^{T-1}B^{\alpha}_t\sum_{j=1}^{T-t}\bigg(\frac{M^{\alpha}_j}{\alpha}\bigg)^{1/2}U^{\alpha}_{j,t}
=\sum_{j=1}^{i-1}\bigg(\frac{M^{\alpha}_j}{\alpha}\bigg)^{1/2}\sum_{t=T-i+1}^{T-j}B^{\alpha}_t U^{\alpha}_{j,t},
\end{align*}
where 
\begin{align*}
B^{\alpha}_t=\frac{(\alpha^{-1}C^{\alpha}_{i,T-i+1})^{1/2}}{\alpha^{-1}\sum_{j=1}^{T-t}C^{\alpha}_{j,t}}
\bigg(\prod_{s=T-i+1}^{t-1}\widehat{f}_s\bigg)\bigg(\prod_{s=t+1}^{T-1}f_s\bigg).
\end{align*}
We will use the facts that $(U^{\alpha}_{1,t})_{t=T-i+1}^{T-1}, (U^{\alpha}_{2,t})_{t=T-i+1}^{T-2}, \dots, U^{\alpha}_{i-1,T-i+1}$ are independent and that each one converges in distribution to a centered normally distributed random vector/variable, and that each $B^{\alpha}_t$ converges a.s.~as $\alpha\to\infty$. A multivariate version of Slutsky's theorem (essentially the continuous mapping theorem for weak convergence) then implies convergence in distribution of \eqref{eq:msep_second_term_rep} to a centered normally distributed random variable.  

Note that 
$$
B^{\alpha}_t\as B_t=\frac{(\lambda_i\E[ZI\{D\leq T-i+1\}])^{1/2}\prod_{s=T-i+1}^{T-1}f_s}{\sum_{j=1}^{T-t}\lambda_j\E[ZI\{D\leq t\}]f_t}.
$$
Note that, for each $j$, as $\alpha\to\infty$, $(U^{\alpha}_{j,t})_{t=T-i+1}^{T-1}$ converges in distribution to a centered normal random vector with covariance matrix $\Sigma$ with 
\begin{align*}
\Sigma_{t,t}&=(f_t-1)^2\E[Z^2I\{D\leq t\}]+\E[Z^2I\{D=t+1\}]\\
&=\sigma^2_t\E[ZI\{D\leq t\}], \\
\Sigma_{t,t+h}&=(f_{t+h}-1)\big((f_t-1)\E[Z^2I\{D\leq t\}]-\E[Z^2I\{D=t+1\}]\big), \quad h>0.
\end{align*}
If $Z$ and $D$ are independent, then it is seen from the above expression that $\Sigma$ is diagonal. In this case, 
\begin{align*}
\sum_{t=T-i+1}^{T-j}B^{\alpha}_t U^{\alpha}_{j,t}\distr \Normal_1\bigg(0,\sum_{t=T-i+1}^{T-j}(B_t)^2\Sigma_{t,t}\bigg)
\end{align*}
and consequently \eqref{eq:msep_second_term_rep} converges in distribution to a centered normally distributed random variable with variance 
\begin{align*}
&\sum_{j=1}^{i-1}\lambda_j\sum_{t=T-i+1}^{T-j}(B_t)^2\Sigma_{t,t}
=\sum_{t=T-i+1}^{T-1}(B_t)^2\Sigma_{t,t}\sum_{j=1}^{T-t}\lambda_j \\
&\quad=\lambda_i\E[ZI\{D\leq T-i+1\}]\prod_{s=T-i+1}^{T-1}f_s^2\sum_{t=T-i+1}^{T-1}\frac{\sigma^2_t/f_t^2}{\sum_{j=1}^{T-t}\lambda_j \E[ZI\{D\leq t\}]}.
\end{align*}
\end{proof}

\begin{proof}[Proof of Theorem \ref{thm:msep_first_term}]
By Corollary \ref{cor:renewal_clt}, $(H^{\alpha},F^{\alpha})\distr (H,F)$, where $H$ and $F$ are independent and normally distributed with zero means and  variances 
$\var(H)=\lambda_i\E[Z^2]\P(D\leq T-i+1)$ and $\var(F)=\lambda_i\E[Z^2]\P(D>T-i+1)$. 
Write $g_{T-i+1}=\frac{\E[ZI\{D\leq T\}]}{\E[ZI\{D\leq T-i+1\}]}=\prod_{s=T-i+1}^{T-1}f_s$ and Note that 
\begin{align*}
&\E\bigg[\frac{(C^{\alpha}_{i,T}-C^{\alpha}_{i,T-i+1}g_{T-i+1})^2}{C^{\alpha}_{T-i+1}} \mid \mathcal{D}^{\alpha}\bigg] \\
&\quad=\frac{\alpha}{C^{\alpha}_{T-i+1}}\E\Big[\big(F^{\alpha}-H^{\alpha}(g_{T-i+1}-1)\big)^2 \mid H^{\alpha}\Big] \\
&\quad=\frac{\alpha}{C^{\alpha}_{T-i+1}}\Big(\E[(F^{\alpha})^2]+(H^{\alpha})^2(g_{T-i+1}-1)^2\Big) \\
&\quad=\frac{\alpha}{C^{\alpha}_{T-i+1}}\Big(\frac{\E[M^{\alpha}_i]}{\alpha}\E[Z^2]\P(D>T-i+1)+(H^{\alpha})^2(g_{T-i+1}-1)^2\Big)
\end{align*}
Since $C^{\alpha}_{i,T-i+1}/\alpha\as \lambda_i\E[Z]\P(D\leq T-i+1)$ and $(H^{\alpha})^2\distr \lambda_i\E[Z^2]\P(D\leq T-i+1)\chi^2_1$, the conclusion follows from Slutsky's theorem. 
\end{proof}

\begin{proof}[Proof of Theorem \ref{thm:msep_third_term}]
$A^{\alpha}_2$ can be expressed as 
$$
A^{\alpha}_2=\alpha^{1/2}\sum_{t=T-i+1}^{T-1}S_t
$$
with $S_t$ as in the proof of Theorem \ref{thm:msep_second_term}. Hence, the arguments in the proof of \ref{thm:msep_second_term} shows that $(A^{\alpha}_2)_{\alpha\geq 0}$ converges in distribution to a normally distributed random variable with zero mean. 
$A^{\alpha}_1$ can be expressed as
$$
A^{\alpha}_1=-\frac{\P(D>T-i+1)}{\P(D\leq T-i+1)}\alpha^{-1/2}\big(C^{\alpha}_{i,T-i+1}-\E[C^{\alpha}_{i,T-i+1}]\big)
$$
from which convergence in distribution to a normally distributed random variable with zero mean follows immediately from Corollary \ref{cor:renewal_clt}. 
Since $(A^{\alpha}_1)_{\alpha\geq 0}$ and $(A^{\alpha}_2)_{\alpha\geq 0}$ are independent, individual convergence in distribution implies joint convergence in distribution. Consequently, mapping theorem for weak convergence implies that the product converges in distribution.  
\end{proof}

\begin{proof}[Proof of Theorem \ref{thm:renewal_clt}]
In order to ease the notation we drop the index $j$ and write $S^{\alpha}=\sum_{k=1}^{M^{\alpha}}X_k$. 
From the renewal process representation of $M^{\alpha}$, there exists an i.i.d.~sequence $(Y_k)$ independent of $(X_k)$ such that the sequence $(T_m)$ given by $T_m=\sum_{k=1}^m Y_k$ satisfies $M^{\alpha}=\sup\{m\geq 1:T_m\leq \alpha\}$. 
Therefore, $\lambda=1/\E[Y]$ and 
\begin{align*}
\alpha^{-1/2}(S^{\alpha}-\E[S^{\alpha}])
&=\alpha^{-1/2}(S^{\alpha}-\E[X]M^{\alpha}+\E[X](M^{\alpha}-\lambda\alpha))+o_{\P}(1)
\end{align*}
using that $\alpha^{-1/2}(\lambda\alpha-\E[M^{\alpha}])=o_{\P}(1)$, i.e.~convergence in probability to zero.  
Using (2.41) in \cite{EKM03}, $\alpha^{-1/2}(M^{\alpha}-\lambda\alpha)=\alpha^{-1/2}(M^{\alpha}-\lambda T_{M^{\alpha}})+o_{\P}(1)$. 
Hence, 
\begin{align*}
\alpha^{-1/2}(S^{\alpha}-\E[S^{\alpha}])
&=\alpha^{-1/2}\big(S^{\alpha}-\lambda\E[X]T_{M^{\alpha}}\big)+o_{\P}(1)\\
&=\alpha^{-1/2}\sum_{k=1}^{M_{\alpha}}(X_k-\lambda \E[X]Y_k)+o_{\P}(1)\\
&=\bigg(\frac{M^{\alpha}}{\alpha}\bigg)^{1/2}(M^{\alpha})^{-1/2}\sum_{k=1}^{M_{\alpha}}(X_k-\lambda \E[X]Y_k)+o_{\P}(1).
\end{align*}
Consequently, 
\begin{align*}
\alpha^{-1/2}(S^{\alpha}-\E[S^{\alpha}])\distr \Normal_T(0,\Sigma),
\end{align*}
where 
\begin{align*}
\Sigma=\lambda\cov(X-\lambda\E[X]Y)=\lambda\cov(X)+\lambda^3\var(Y)\E[X]\E[X]^{\trans}.
\end{align*}
If $M^{\alpha}$ is Poisson distributed, then $\var(Y)=1/\lambda^2$ and hence $\Sigma=\lambda\E[XX^{\trans}]$ is diagonal with $\Sigma_{t,t}=\lambda\E[Z^2I\{D=t\}]$.
\end{proof}

\end{document}